\documentclass[letterpaper,10pt,conference]{arxiv-conference} \IEEEoverridecommandlockouts\usepackage{amsmath,amssymb,amsfonts} \usepackage{graphicx} \usepackage{cite} \usepackage{url} \usepackage[hidelinks,bookmarksnumbered=true, pdfauthor={Alexey Peregudin and Ngoc Tuan Dinh}, pdftitle={Exact Noise Limits for Bounded Scalar Stabilization Experiments}]{hyperref} \newtheorem{theorem}{Theorem} \newtheorem{lemma}{Lemma}  \newtheorem{corollary}{Corollary} \newcommand{\R}{\mathbb R} \newcommand{\norm}[1]{\left\|#1\right\|} \newcommand{\inner}[2]{\langle#1,#2\rangle} \newcommand{\dd}{\,d} \newcommand{\eps}{\varepsilon} \newcommand{\Proof}{\par
\noindent\emph{Proof:}\ } \newcommand{\Endproof}{\hfill$\square$\par
} \newcommand{\statementdisplayend}{\vspace{-0.6\baselineskip}}  \clubpenalty=10000 \widowpenalty=10000 \par
\title{Exact Noise Limits for Bounded Scalar Stabilization Experiments} \author{Alexey Peregudin and Ngoc Tuan Dinh\thanks{Alexey Peregudin is with the School of Electrical and Electronic Engineering, University of Sheffield, Sheffield, United Kingdom (e-mail: a.peregudin@sheffield.ac.uk).}\thanks{Ngoc Tuan Dinh is with ITMO University, St. Petersburg, Russia (e-mail: dinhngoctuan6789@gmail.com).}} \hbadness=10000\vbadness=10000
\begin{document} \maketitle\thispagestyle{empty} \begin{abstract} How much unknown process noise can a bounded experiment tolerate while still establishing that the plant can be stabilized? We answer this question for scalar discrete-time systems with bounded inputs, exact state measurements, and a disturbance-energy budget proportional to the experiment length. We distinguish individual stabilizability of every consistent model, stabilization by one common gain, and a common quadratic certificate. We determine their exact noise ceilings for every drift. The last two requirements coincide; individual stabilizability generally tolerates more noise. Bounded periodic inputs approach the ceilings, while adversarial disturbances prevent their attainment. We also determine the long-horizon limits. For unstable plants, a short experiment can establish individual stabilizability at noise levels where every sufficiently long experiment fails. The same obstruction yields upper bounds for multidimensional systems with real eigenvalues. Input design may use plant knowledge, but certification uses only the recorded data and noise bound. The results therefore provide fundamental robustness benchmarks for unknown-plant designs. \end{abstract}  \par
\section{Introduction} \label{sec:introduction} Input limits constrain how much information an experiment can provide about an uncertain plant. With unknown disturbances, the data describe a set of models. Every model may be individually stabilizable even though no single controller stabilizes them all. This letter determines exactly how much process noise a bounded scalar experiment can tolerate under these different requirements. \par
Worst-case input design has a substantial history. Optimal identification under bounded disturbances was studied in \cite{tse1993}, and set-membership experiment design has been developed for conditional identification \cite{casini2006} and constrained linear systems \cite{tanaskovic2014}. These works optimize identification error or the size of a feasible model set. Control-oriented identification instead relates experimental cost to the eventual control requirement; see, for example, \cite{bombois2006}. Our objective is directly a stabilization property of every model consistent with the experiment. Finite-alphabet stationary input design also provides established ways to enforce amplitude constraints \cite{valenzuela2013}. \par
The informativity framework formalizes this viewpoint \cite{informativity2020}. Data-based controller parametrizations \cite{depersis2020}, the matrix S-lemma \cite{slemma2022}, and Petersen's lemma \cite{petersen2022} provide powerful ways to certify stabilization from a given noisy record. In particular, exact tests for common quadratic stabilization are available under their respective assumptions. The distinction between stabilizability, stabilization, and common Lyapunov certificates is studied explicitly in \cite{commonlyapunov2022}. The question here includes a further design quantifier: for a specified plant, does some amplitude-bounded input produce a successful record for \emph{every} admissible disturbance, and what is the largest noise level for which this is possible? \par
The noise model is essential: disturbance energy is bounded by $\varepsilon^2T$, without probabilistic or samplewise assumptions. The consequences of replacing instantaneous bounds by energy bounds were examined in \cite{tradeoffs2021}. Additional samples need not average adversarial noise away, so we ask for a noise ceiling. Online noise-free experiment design \cite{beyond2022}, universal inputs using controllability or stabilizability priors \cite{shakouri2025}, and statistical stabilization sample complexity \cite{tsiamis2022} address different information patterns. \par
Our purpose is to establish an \emph{oracle limit on experimental robustness}. Imagine a designer who knows the true plant and chooses the best predetermined bounded input. The resulting data must still establish the required property for every consistent model: the certificate is not told which model is true. Plant knowledge therefore helps choose the experiment, but does not replace evidence in the record. If even this designer cannot succeed, no predetermined unknown-plant design can do better. A matching construction identifies the robustness that implementable designs could aim to approach. Finding the input without plant knowledge is a further design problem, outside this letter's benchmark. \par
We give closed-form ceilings for individual stabilizability, common-gain stabilization, and common quadratic certification, for every scalar drift. The latter two coincide, whereas the first can tolerate arbitrarily more noise in relative terms. Explicit disturbances prove impossibility at and above each ceiling. Saturated periodic inputs attain every strictly smaller noise level. We also determine the limits along all integer horizons. The common thresholds converge to their ceiling, but for unstable plants the individual limit is strictly lower: a short experiment can tolerate noise that defeats every sufficiently long one. Finally, the scalar adverse trajectory yields a multidimensional upper bound through any real left eigenvector, without a dimension factor. A separation argument and a one-step-difference energy certificate connect the converse to bounded periodic input design. All proofs are included. \par
\clearpage\section{Problem and Exact Noise Limits} \label{sec:problem} Consider the scalar system \begin{equation} x_{k+1}=ax_k+bu_k+w_k,\qquad x_0=0,\quad b\ne0. \label{eq:system} \end{equation} We assume the plant can be reset to the origin before the experiment. An experiment of length $T$ uses a predetermined input with $|u_k|\le\bar u$, where $\bar u>0$, and records all states without measurement error. The unknown process disturbance satisfies $\norm{w}_2\le\eps\sqrt T$, where $\eps\ge0$ and $w=(w_0,\ldots,w_{T-1})^\top$. Thus $\eps$ is a bound on the root-mean-square disturbance, not on each sample. The same bound defines the consistent real pairs $(\alpha,\beta)\in\R^2$: \begin{equation} \Sigma_\eps=\Bigl\{(\alpha,\beta):\sum_{k=0}^{T-1} \!(x_{k+1}\!-\alpha x_k\!-\beta u_k)^2\le\eps^2T\Bigr\}. \label{eq:consistency} \end{equation} This set contains $(a,b)$ and is therefore nonempty. Alternative models may have $\beta=0$: controllability of the true plant is not supplied as a prior to the certificate. \par
We distinguish three requirements on this realized record: \begin{enumerate} \item\emph{Individual stabilizability:} every $(\alpha,\beta)\in\Sigma_\eps$ admits a gain $K_{\alpha,\beta}$ with $|\alpha+\beta K_{\alpha,\beta}|<1$. \item\emph{Common gain:} one $K$ satisfies $|\alpha+\beta K|<1$ for every $(\alpha,\beta)\in\Sigma_\eps$. \item\emph{Common quadratic certificate:} there are $K$ and $p>0$ such that $p(\alpha+\beta K)^2-p<0$ for every $(\alpha,\beta)\in\Sigma_\eps$. \end{enumerate} The first condition does not supply a controller valid for all consistent models. For a scalar model, failure of individual stabilizability means precisely $\beta=0$ and $|\alpha|\ge1$. The last two coincide in one state dimension: $p=1$ certifies every common stabilizing gain. We keep them distinct because, in higher dimensions, a common gain need not admit a common quadratic certificate \cite{commonlyapunov2022}. The amplitude bound applies to the experiment; the certified feedback is unconstrained and stabilizes the disturbance-free models. \par
An input is \emph{successful at level $i$} if the corresponding requirement holds for every admissible disturbance. Let $\eps_i(T)$ be the supremum of the noise levels admitting a successful length-$T$ input, and set \begin{equation} \eps_i^\star=\sup_{T\ge1}\eps_i(T). \label{eq:ceiling} \end{equation} For an empty set of successful noise levels we use supremum zero. Input design may use $(a,b)$, but the input is fixed before the disturbance is chosen. The gain may depend on the record. Certification uses only the data and noise bound, retaining every model in \eqref{eq:consistency}. It uses neither the true parameters nor information inferred from the input-selection rule. This is an oracle benchmark for predetermined experiments. Adaptive experiments have different quantifiers and are outside our scope. \par
\begin{theorem}[Exact scalar noise limits] \label{thm:main} For every $a\in\R$, $b\ne0$, and $\bar u>0$, \begin{align} &\eps_1^\star=\begin{cases} \displaystyle\frac{|b|\bar u}{1+\sqrt{(1+a^2)/2}},& |a|<1,\\[1mm] |b|\bar u/2,& |a|\ge1. \end{cases}\label{eq:individual-ceiling}\displaybreak[1]\\ &\eps_2^\star=\eps_3^\star=\frac{|b|\bar u}{1+\max\{1,|a|\}}.\label{eq:common-ceiling} \end{align} Every noise strictly below its displayed ceiling admits a finite successful experiment. At or above the ceiling, every input at every horizon fails. Moreover, the common-certificate thresholds converge along all integer horizons: \begin{equation} \lim_{T\to\infty}\eps_2(T)=\lim_{T\to\infty}\eps_3(T) =\frac{|b|\bar u}{1+\max\{1,|a|\}}. \label{eq:limit} \end{equation} \statementdisplayend\end{theorem} \par
\begin{figure}[t] \centering\includegraphics[width=\columnwidth]{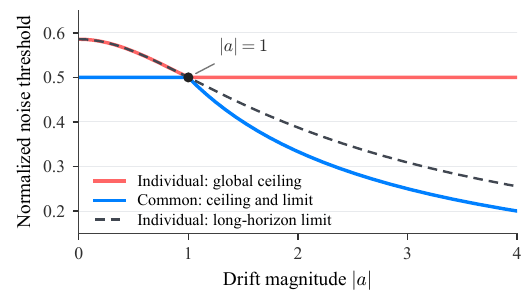} \caption{Exact noise thresholds, normalized by $|b|\bar u$. The solid curves are the global ceilings; the blue curve is also the common-certificate limit. The dashed curve is the individual long-horizon limit (Theorem~\ref{thm:individual-limit}). For $|a|>1$, optimizing the experiment length gives a strictly higher individual ceiling than using arbitrarily long experiments.} \label{fig:ceilings} \end{figure} \par
The three ceilings agree only at $a=\pm1$. At every other drift, individual stabilizability tolerates strictly more noise. For $|a|\ge1$, a single sample achieves every subcritical level-1 noise, and $\eps_1^\star/\eps_2^\star=(1+|a|)/2$ is unbounded. The individual long-horizon limit can be smaller than its global ceiling; Section~\ref{sec:long-horizon} identifies it and carries its obstruction into multidimensional systems. \par
\emph{Proof outline.} First, we show how a disturbance can hide a bad model or prevent the consistent models from sharing a gain. These constructions explain the ceilings and their unattainable endpoints. To prove achievability, a quadratic energy test separates the cost of creating a trajectory from the cost of fitting a bad explanation. Periodic inputs make this test a frequency calculation with only bounded endpoint errors. We treat the simple individual requirement first, then balance all directions needed by a common gain. The same estimate also yields the common limit over all integer horizons. A final closed-trajectory construction gives the individual limit and embeds the same adverse scalar record in a real mode of a larger plant. \par
\section{Record Geometry and Impossibility} \label{sec:converse} The upper bounds must defeat every input. The key is that the adversary can spend noise both on changing the observed trajectory and on making an incompatible model fit that trajectory. We normalize states and disturbances by $|b|\bar u$ and set $\delta=\eps/(|b|\bar u)$, giving \begin{equation} x_{k+1}=ax_k+q_k+v_k,\quad|q_k|\le1,\quad\norm{v}_2\le\delta\sqrt T. \label{eq:normalized} \end{equation} Every such $q$ is implemented by $u_k=\bar u\,\operatorname{sgn}(b)q_k$. This invertible change of variables preserves all three requirements. In this section $x=(x_0,\ldots,x_{T-1})^\top$, $y=(x_1,\ldots,x_T)^\top$, and $q=(q_0,\ldots,q_{T-1})^\top$. All unmarked norms are Euclidean. For $d\ge0$, write $\Sigma_d=\{(\alpha,\beta)\in\R^2:\norm{y-\alpha x-\beta q}\le d\}$. \par
\begin{lemma}[A scalar separation test] \label{lem:separation} If $\Sigma_d$ is nonempty, it has a common gain if \begin{equation} \min_{|\lambda|\le1}\norm{x-\lambda y}>d. \label{eq:separation} \end{equation} If $y\in\operatorname{span}\{x,q\}$, this condition is also necessary. \end{lemma} \Proof For sufficiency, let $h=x-\lambda_*y$ at a minimum in \eqref{eq:separation}. Projection onto the segment gives $(\lambda-\lambda_*)h^\top y\le0$ for all $|\lambda|\le1$. Hence $\lambda_*h^\top y=|h^\top y|$ and $h^\top x=\norm h^2+|h^\top y|>0$. Choose $K=h^\top q/(h^\top x)$. For every consistent pair, $y=\alpha x+\beta q+e$ with $\norm e\le d$ gives \[ |\alpha+\beta K| \le\frac{|h^\top y|+d\norm h}{|h^\top y|+\norm h^2}<1. \] For necessity, let $|\lambda|\le1$ satisfy $\norm{x-\lambda y}\le d$ and write $y=a_0x+b_0q$. The two models \[ (\alpha_\pm,\beta_\pm) =\bigl(a_0\pm(1-\lambda a_0),\ b_0(1\mp\lambda)\bigr) \] have residuals $\mp(x-\lambda y)$, so both are consistent. For a gain $K$, put $t=a_0+b_0K$. The closed-loop coefficients $1+(1-\lambda)t$ and $-1+(1+\lambda)t$ cannot both be stable: for $-1<\lambda<1$ this requires both $t<0$ and $t>0$; at $\lambda=\pm1$ one coefficient is fixed at $\pm1$. \Endproof\par
We first prove the upper bound in \eqref{eq:common-ceiling}. For $|a|\le1$ choose $v=-q/2$. Then $y=ax+q/2$ and \begin{equation} \norm{x-ay}^2 =\tfrac14\norm q^2-(1-a^2)x_T^2\le T/4. \label{eq:common-upper} \end{equation} For $|a|>1$ choose $v=-q/(1+|a|)$ and $\lambda=1/a$. Since $x-\lambda y=-(q+v)/a$, its norm is at most $\sqrt T/(1+|a|)$. These proportional disturbances put $y$ in $\operatorname{span}\{x,q\}$, so Lemma~\ref{lem:separation} proves failure at the claimed thresholds, including zero input and dependent state and input rows. \par
Both individual upper bounds use one identity. For any trajectory $z_0,\ldots,z_T$, let $\xi,\eta$ denote its past and next state rows. Expansion gives \begin{equation} \begin{split} \norm{\eta-a\xi}^2={}& \tfrac{(1+a)^2}{4}\norm{\eta-\xi}^2 +\tfrac{(1-a)^2}{4}\norm{\eta+\xi}^2\\ &{}+\tfrac{1-a^2}{2}(z_T^2-z_0^2). \end{split} \label{eq:endpoint} \end{equation} The two weights sum to $s^2$, where $s=\sqrt{(1+a^2)/2}$. Thus generation energy controls the cost of fitting a unit-drift model, up to an endpoint term. \par
For $|a|\le1$, take the noiseless response to $q$, with $z_0=0$, and put $Q=\norm q$ and $R=\min_{\sigma=\pm1}\norm{\eta-\sigma\xi}$. The endpoint term is nonnegative, so \eqref{eq:endpoint} gives $Q^2\ge s^2R^2$. For $Q>0$, choose a minimizing $\sigma$ and use $v=-\gamma q$, where $\gamma=R/(Q+R)$. Its norm and the residual of the unstabilizable alternative $(\sigma,0)$ are both $QR/(Q+R)$, which is at most $\sqrt T/(1+s)$. If $Q=0$, the zero record already admits $(1,0)$. For $|a|\ge1$, $v=-q/2$ makes the unstabilizable pair $(a,0)$ consistent. The same counterexamples remain valid at larger noise budgets, proving all endpoint assertions. \par
\par
\section{Experiments That Reach Every Subcritical Noise Level} \label{sec:construction} We now construct inputs that work against \emph{every} disturbance. The main device measures the combined cost of producing a trajectory and making it fail the separation test. Keeping the dynamics in difference form makes this cost easy to estimate on periodic inputs: only the first and last samples distinguish a finite record from a cycle. \par
\subsection{A certificate using one-step differences} Let $S$ be the backward shift on $\R^T$, with $(Sz)_0=0$ and $(Sz)_k=z_{k-1}$. Since $x=Sy$, set \[ F=I-aS,\quad G=S-\lambda I,\quad v=Fy-q,\quad x-\lambda y=Gy. \] Both operators involve only one-step differences. For $c>0$, define \begin{equation} J_{c,\lambda,T}(q) =\min_y\bigl\{\norm{Fy-q}^2+c\norm{Gy}^2\bigr\}. \label{eq:energy} \end{equation} If both the true disturbance and a separation obstruction have norm at most $\delta\sqrt T$, then $J_{c,\lambda,T}(q)\le(1+c)\delta^2T$. Thus \begin{equation} J_{c,\lambda,T}(q)>(1+c)\delta^2T \quad\text{for all }|\lambda|\le1 \label{eq:sufficient} \end{equation} guarantees a common gain. This test accounts for both noise budgets. \par
The two residuals cannot be chosen independently. Writing $e=Gy$, commutation $FG=GF$ gives $Fe-Gv=Gq$. Pair with $h$ and apply Young's inequality to the two terms. Minimizing over $y$ bounds the energy without solving for the trajectory, for any $h\in\R^T$: \begin{equation} J_{c,\lambda,T}(q)\!\ge\!2\inner{Gq}{h}-\norm{G^\top h}^2 -c^{-1}\norm{F^\top h}^2. \label{eq:dual} \end{equation} A suitable periodic $h$ makes this bound sharp apart from endpoint terms. \par
For an input of period $p$, define the Fourier coefficients \[ \widehat q_j=\frac1p\sum_{k=0}^{p-1}q_ke^{-2\pi i jk/p}, \quad0\le j<p. \] Its power measure $\mu$ places mass $|\widehat q_j|^2$ at $\omega_j=2\pi j/p$. Each integral below is therefore a finite weighted sum; if $q_k\in\{-1,1\}$, its total mass is one. Set $D_t(\omega)=1+t^2-2t\cos\omega$. \par
\begin{lemma}[Periodic inputs on finite records] \label{lem:periodic} Fix $a\in\R$, $c>0$, and a real periodic input $q$. Suppose $D_a$ is positive at every frequency with nonzero input power. Then its length-$T$ prefix satisfies \begin{equation} J_{c,\lambda,T}(q)\ge T\int\frac{cD_\lambda}{D_a+cD_\lambda}\dd\mu-C, \quad|\lambda|\le1, \label{eq:periodic} \end{equation} where $C$ is independent of $T$ and $\lambda$. \end{lemma} \Proof On one period choose the dual vector with Fourier coefficients \[ \widehat h_j= \frac{c(e^{-i\omega_j}-\lambda)} {D_a(\omega_j)+cD_\lambda(\omega_j)}\,\widehat q_j, \] setting inactive coefficients to zero. In the cyclic version of \eqref{eq:dual}, Parseval's identity gives the integral in \eqref{eq:periodic} per sample. There are finitely many active frequencies and $D_a>0$ at each, so all entries of this periodic $h$ are bounded uniformly in $\lambda$. Repeat $h$ and take its first $T$ entries. Relative to the cyclic expressions, $Gq$ changes only in its first coordinate, while $F^\top h$ and $G^\top h$ change only in their last coordinates. Their contribution to \eqref{eq:dual} is bounded independently of $T,\lambda$. The incomplete last period also has bounded length and bounded contribution. This proves the assertion. \Endproof The hypothesis is automatic for $a\ne\pm1$. At $a=1$ it requires zero mean; at $a=-1$ it excludes power at $\pi$. In particular, the lemma covers every integer horizon, including incomplete periods. \par
\subsection{Individual stabilizability: a four-sample pattern} To ensure individual stabilizability, it is enough to exclude unstable alternatives with zero input coefficient. For any real $a$, repeat $(1,1,-1,-1)$ and put \[ s=\sqrt{(1+a^2)/2},\qquad\delta_1=(1+s)^{-1}. \] At the two active frequencies $D_a=2s^2$ and $D_\lambda=1+\lambda^2$. The function $f(t)=(1+t)/(2s+1+t)$ is concave on $[0,1]$. Its endpoint values satisfy $f(0)\ge(1+s)^{-2}$ and $f(1)=(1+s)^{-1}$, so its chord gives, for $t=\lambda^2$, \[ \frac{1+t}{2s+1+t} \ge(1-t)f(0)+tf(1)\ge\frac{1+st}{(1+s)^2}. \] Lemma~\ref{lem:periodic}, with $c=s$, therefore gives $J_{s,\lambda,T}(q)\ge T(1+s\lambda^2)\delta_1^2-C$. An unstabilizable alternative has $\beta=0$, $|\alpha|\ge1$. With $\lambda=1/\alpha$, its residual bound implies $\norm{x-\lambda y}\le|\lambda|\delta\sqrt T$, hence $J_{s,\lambda,T}(q)\le T(1+s\lambda^2)\delta^2$. For $\delta<\delta_1$, these bounds contradict one another whenever $T>C/(\delta_1^2-\delta^2)$, uniformly in $\lambda$. This proves an all-integer-horizon lower bound for every drift, which we will also use in Section~\ref{sec:long-horizon}. For $|a|<1$, it reaches the global individual ceiling. \par
For $|a|\ge1$, a single sample is enough below $1/2$: the input $q_0=1$ gives $|x_1|\ge1-\delta>\delta$. Because $x_0=0$, no zero-channel model can fit that record. This completes the lower bound in \eqref{eq:individual-ceiling}. \par
\subsection{Common certification: distributing input power} A common gain must handle all separation directions $\lambda$ simultaneously. Lemma~\ref{lem:periodic} turns this requirement into a choice of input power across frequencies. The next lemma supplies that choice using only inputs at the amplitude limits. We first balance power for a stable reference drift, then compare every plant with such a reference. A small amount of slack below the ceiling lets one construction cover stable, unstable, and unit drifts. \par
\begin{lemma}[Bounded spectral design] \label{lem:design} For every $0\le r<1$ and $\kappa>0$, there is a deterministic periodic input $q_k\in\{-1,1\}$ with zero mean over its period whose power measure satisfies, for $|\lambda|\le1$ and $0\le\tau\le1$, \begin{equation} \int\frac{\tau D_\lambda}{D_r+\tau D_\lambda}\dd\mu\ge\frac{\tau}{1+\tau}-\kappa. \label{eq:spectral-design} \end{equation} \statementdisplayend\end{lemma} The construction and proof are in the Appendix. Its idea is to generate a family of binary correlations by repeatedly reading and reversing one of finitely many signs. Mixing two adjacent family members balances the critical direction. Finite word counts then realize the resulting power distribution to any required accuracy. Randomness is only a device for describing those counts: the final periodic input is fixed before the disturbance is chosen. Its period need not be short, which is consistent with the benchmark's optimization over all experiment lengths. \par
\subsection{One construction for every drift} It suffices to treat $a\ge0$: the transformation $x_k\mapsto(-1)^kx_k$, $(q_k,v_k)\mapsto(-1)^{k+1}(q_k,v_k)$ reverses the drift without changing either budget or any certification requirement. Fix $\delta<1/(1+\max\{1,a\})$. Choose $M>\max\{1,a\}$ close enough that $\delta<1/(1+M)$. The stable reference $r=a/M^2<1$ satisfies the pointwise comparison \begin{equation} M^2D_r-D_a=(M^2-1)(1-a^2/M^2)>0. \label{eq:drift-comparison} \end{equation} Replacing $D_a$ by this larger quantity lowers the spectral energy, so a design for the reference also certifies the actual plant. \par
Choose $\kappa>0$ such that $\kappa<1/(1+M)-(1+M)\delta^2$, and take the input from Lemma~\ref{lem:design}. At $a=1$, its exact zero mean removes the only zero of $D_a$ from the active frequencies. Thus Lemmas~\ref{lem:periodic}--\ref{lem:design}, with $c=M$ and $\tau=1/M$, give for all $|\lambda|\le1$ \begin{align*} J_{M,\lambda,T}(q) &\ge T\int\frac{MD_\lambda}{D_a+MD_\lambda}\dd\mu-C\\ &\ge T\int\frac{M^{-1}D_\lambda}{D_r+M^{-1}D_\lambda}\dd\mu-C\\ &\ge T\bigl(1/(1+M)-\kappa\bigr)-C. \end{align*} The coefficient of $T$ exceeds $(1+M)\delta^2$ by a fixed positive amount. For sufficiently large $T$, it absorbs the bounded endpoint loss $C$, proving \eqref{eq:sufficient}. \par
For every fixed subcritical common-certification noise, we have chosen one periodic input whose prefixes work at \emph{all sufficiently large integer horizons}. Hence the normalized thresholds have liminf at least $1/(1+\max\{1,|a|\})$. The all-horizon converse gives the matching limsup. Restoring the scale $|b|\bar u$ proves \eqref{eq:limit} and completes Theorem~\ref{thm:main}. \par
\section{Long Experiments and Multidimensional Obstructions} \label{sec:long-horizon} The common ceiling describes both the best experiment length and the limit of long experiments. Individual stabilizability behaves differently. We first complete its asymptotic picture, then show that the same obstruction survives in a real mode of a multidimensional plant. \par
\begin{theorem}[The individual long-horizon limit] \label{thm:individual-limit} For every $a\in\R$, $b\ne0$, and $\bar u>0$, \begin{equation} \lim_{T\to\infty}\eps_1(T) =\frac{|b|\bar u}{1+\sqrt{(1+a^2)/2}}. \label{eq:individual-limit} \end{equation} \statementdisplayend\end{theorem} \Proof The four-sample construction in Section~\ref{sec:construction}-B gives the lower bound for every $a$. For the converse, use the normalized coordinates \eqref{eq:normalized} and write $s=\sqrt{(1+a^2)/2}$ and $\delta_1=(1+s)^{-1}$. The idea is to close the trajectory at both ends. This removes the terminal-state term that prevented the earlier stable-drift converse from applying to unstable plants. \par
Suppose $|a|\ne1$. For any $|q_k|\le1$, there is a sequence $z_0=z_T=0$ whose forcing $f_k=z_{k+1}-az_k$ differs from $q$ in one coordinate only: $f=q+he_j$, with $|h|$ bounded independently of $T$ and $q$. For $|a|<1$, run the recursion forward from zero and reset its last state to zero; then $|h|\le1/(1-|a|)$. For $|a|>1$, run it backward from $z_T=0$ and reset its initial state to zero; then $|h|\le|a|/(|a|-1)$. \par
The endpoint term in \eqref{eq:endpoint} now vanishes for every $a$. Choose $\sigma\in\{-1,1\}$ minimizing \[ R^2:=\sum_{k=0}^{T-1}(z_{k+1}-\sigma z_k)^2 \le\frac{\norm f^2}{s^2}. \] Use the actual trajectory $x_k=(1-\delta_1)z_k$. Its true disturbance is $v=-\delta_1q+(1-\delta_1)he_j$; the residual of the unstabilizable alternative $(\sigma,0)$ has norm $(1-\delta_1)R\le\delta_1\norm f$. This scale balances the two costs. Since the correction occupies one coordinate and \mbox{$\norm f^2\le T+2|h|+h^2$}, \[ \max\bigl\{\norm v^2,\norm{y-\sigma x}^2\bigr\} \le\delta_1^2T+O_a(1), \] with a remainder uniform over inputs. Thus every $\delta>\delta_1$ defeats every sufficiently long experiment. At $a=\pm1$, half cancellation gives the upper bound $\delta_1=1/2$ at every horizon. Restoring $|b|\bar u$ proves the limit. \Endproof\par
For $|a|>1$, the individual limit is strictly below its global ceiling, as Fig.~\ref{fig:ceilings} shows. We discuss the significance of this gap in Section~\ref{sec:conclusion}. \par
To carry the obstruction beyond one state, consider $x_{k+1}=Ax_k+Bu_k+w_k$, with $x_0=0$, $x_k\in\R^n$, $u_k\in\R^m$, and $\norm{u_k}_\infty\le\bar u$. Bound the matrix $W=[w_0\ \cdots\ w_{T-1}]$ by $\norm W_2\le\eps\sqrt T$, where $\norm{\cdot}_2$ is the spectral norm, and use the same norm and radius for consistent-model residuals. Define $\eps_i(T;A,B)$ by the same experiment quantifiers and three requirements, using Schur stability and a common positive definite quadratic Lyapunov matrix in place of scalar stability and $p$. All real matrix pairs remain eligible as alternative models. \par
\begin{corollary}[An obstruction in a real mode] \label{cor:real-mode} Let $\ell A=a\ell$ for a real row $\ell$ with $\norm\ell_2=1$. Then, for $i=1,2,3$, \begin{equation} \limsup_{T\to\infty}\eps_i(T;A,B) \le\frac{\bar u\norm{\ell B}_1}{1+\sqrt{(1+a^2)/2}}. \label{eq:real-mode} \end{equation} \statementdisplayend\end{corollary} \Proof The projected state $\ell x_k$ has drift $a$ and forcing bounded by $\beta_\ell=\bar u\norm{\ell B}_1$. For $\beta_\ell>0$, apply the preceding scalar converse to that forcing and inject its disturbance along $\ell^\top$. With $P=\ell^\top\ell$, choose the alternative \[ A'=A+(\sigma-a)P,\qquad B'=(I-P)B. \] It is unstabilizable: $\ell A'=\sigma\ell$, $\ell B'=0$, and $|\sigma|=1$. Its residual at each sample is $\ell^\top(\ell x_{k+1}-\sigma\ell x_k)$. Both full disturbance and residual matrices therefore have exactly their scalar norms. The scalar upper bound defeats level 1 and hence all three levels. If $\beta_\ell=0$, take $\sigma=1$: the zero-noise record has $\ell x_k=0$ and already admits this alternative. \Endproof\par
The factor $\bar u\norm{\ell B}_1$ is the largest input forcing available in the mode. Taking the smallest bound over real left eigenvectors gives a plant-level obstruction, including for diagonal systems. No dimension factor is lost, and the rank-one proof also holds with Frobenius disturbance and residual norms. This is an upper bound, rather than an exact multidimensional limit; it makes no assertion for plants without real eigenvalues and does not identify common gains with common quadratic certificates. \par
\section{Conclusion and Future Work} \label{sec:conclusion} The global ceiling and the long-horizon limit answer different experimental questions: can \emph{some} experiment length succeed, and what remains possible when the experiment is long? For common certification the answers agree. For individual stabilizability of an unstable plant, they differ: noise levels between the two thresholds allow a successful one-sample experiment but defeat every sufficiently long experiment. Longer records also permit more total disturbance energy, since the budget is $\eps^2T$. Thus additional samples alone do not ensure a stronger worst-case guarantee. At a fixed absolute noise budget, appending data can only reduce the set of consistent models; there is no conflict with that basic information principle. \par
The oracle ceilings give a quantitative target for designs that must choose inputs without knowing the plant. An important next question is how much robustness is lost when that knowledge is unavailable. Even below a ceiling, the required experiment may be long: useful finite-horizon bounds and practical constructions near the minimum successful length would complement the exact noise limits. For multidimensional systems, Corollary~\ref{cor:real-mode} identifies an unavoidable obstruction, but matching constructions and the gap between common gains and common quadratic certificates remain to be understood. State constraints and measurement errors are further steps toward experiments that can be implemented on a physical plant. \par
\section*{Appendix: A Binary Spectral Construction} \label{sec:appendix} We prove Lemma~\ref{lem:design} by choosing an auxiliary power measure with $\int D_r/D_\lambda\dd\mu\le1$. The arithmetic--harmonic mean inequality then gives the required energy bound. Finally, we realize the measure to arbitrary accuracy by a deterministic periodic input. \par
\emph{A family with explicit correlations.} Start with $N\ge1$ independent uniform signs. At each step choose one uniformly, record its value $\zeta_k$, and reverse it. The uniform distribution of the signs is invariant. With $\rho_N=1-2/N$, conditioning on the selected sign gives \[ C_0=1,\qquad C_j:=\mathbb E[\zeta_0\zeta_j]=-N^{-1}\rho_N^{j-1} \quad(j\ge1). \] Duplicate each value and randomize the two time phases. The resulting stationary binary process has power measure $\mu_N$ with $\int\cos(j\omega)\dd\mu_N=R_j$, where $R_{2j}=C_j$ and $R_{2j+1}=(C_j+C_{j+1})/2$. \par
For $|z|<1$, expand \[ D_z^{-1}=\frac{1+2\sum_{j\ge1}z^j\cos(j\omega)}{1-z^2}. \] Summing the geometric series gives \begin{equation} g_N(z):=\int D_z^{-1}\dd\mu_N =\frac{N+(N-1)z}{N+(2-N)z^2}. \label{eq:binary-g} \end{equation} For $N=1$ this is the period-four input $(1,1,-1,-1)$. \par
\emph{Balancing the critical direction.} We mix two neighboring family members so that the ratio integral touches one at $\lambda=r$ and never exceeds it. Write $A_t(z)=t+(2-t)z^2$ and set \[ \nu=1+\frac{2r}{1-r^2},\qquad N=\lfloor\nu\rfloor,\quad\theta=\nu-N. \] The desired mixture is \begin{equation} \mu_*= \frac{(1-\theta)A_N(r)\mu_N+\theta A_{N+1}(r)\mu_{N+1}} {(1+r)^2}. \label{eq:binary-mixture} \end{equation} Its weights are nonnegative and sum to one because $A_t$ is affine in $t$ and $A_\nu(r)=(1+r)^2$. Writing $g(z)=\int D_z^{-1}\dd\mu_*$, substitution gives $g(r)=1/(1-r^2)$. \par
For $0<|z|<1$, the identity $D_r=(r/z)D_z+(1-r/z)(1-rz)$ gives \[ \int\frac{D_r}{D_z}\dd\mu_*-1 =\frac{z-r}{z}\bigl((1-rz)g(z)-1\bigr). \] Substituting \eqref{eq:binary-g}--\eqref{eq:binary-mixture} and factoring the right-hand side yields \begin{equation} \int\frac{D_r}{D_z}\dd\mu_* =1-\frac{1-r}{1+r} \frac{(z-r)^2B(z)}{A_N(z)A_{N+1}(z)}, \label{eq:binary-ratio} \end{equation} where $B(z)=A_{N+1-\theta}(z)-\theta(1-\theta)(1+z)^2$. For $t\ge1$ and $|z|\le1$, $A_t(z)=1+z^2+(t-1)(1-z^2)\ge1+z^2$. Since $0\le\theta<1$, this also gives \[ B(z)\ge1+z^2-\tfrac14(1+z)^2 \ge\tfrac12(1+z^2)>0. \] Thus the subtracted term in \eqref{eq:binary-ratio} is nonnegative: the ratio integral is at most one for $|z|<1$, including $z=0$ by continuity. At $r=0$ the same construction selects $\mu_1$. \par
For $\tau>0$ and $|\lambda|<1$, apply the arithmetic--harmonic mean inequality to $\tau+D_r/D_\lambda$: \[ \int\frac{\tau D_\lambda}{D_r+\tau D_\lambda}\dd\mu_* \ge\frac{\tau}{\tau+\int D_r/D_\lambda\dd\mu_*} \ge\frac{\tau}{1+\tau}. \] The original integrand is bounded and continuous since $D_r>0$, so passage to $\lambda=\pm1$ gives the endpoints. At $\tau=0$ the assertion is immediate. \par
\emph{A deterministic periodic realization.} It remains to remove the auxiliary randomness. Fix a word length $m$ and approximate the probabilities of all binary $m$-words of the stationary mixture by rational probabilities. Preserve equal weights for a word and its negative; this is possible because the process is invariant under reversal of all signs. Concatenate copies in those integer proportions and repeat the resulting word periodically. Its mean is exactly zero. At a fixed lag $j<m$, at most a fraction $j/m$ of products crosses word boundaries; the remaining correlation approaches that of the mixture as the rational approximation improves. Letting $m$ grow and the rational errors tend to zero, the periodic correlations approach those of $\mu_*$ at every fixed lag. Trigonometric-polynomial approximation gives convergence of power integrals for every continuous function of frequency. \par
For fixed $r<1$, the kernel in \eqref{eq:spectral-design} is jointly continuous in $(\omega,\lambda,\tau)$ on $[0,2\pi]\times[-1,1]\times[0,1]$, since $D_r\ge(1-r)^2>0$. A finite uniform net in $(\lambda,\tau)$ therefore makes the convergence uniform in both parameters. Choose the word accurately enough that the error is below $\kappa$. It is a single deterministic input, fixed for all disturbances. This proves Lemma~\ref{lem:design}. \Endproof\par
\section*{Acknowledgment} \emph{AI use statement:} This letter was written with the assistance of ChatGPT 5.6 Sol. Working with AI was an iterative process, making it difficult to isolate its specific contribution. \space\space \par
\par
\end{document}